\documentclass[12pt,american,refpage,intoc,bibliography=totoc,index=totoc,BCOR=7.5mm,captions=tableheading]{article}
\usepackage{mathptmx}
\usepackage[T1]{fontenc}
\usepackage[latin9]{inputenc}
\usepackage[a4paper]{geometry}
\usepackage{color}
\usepackage{babel}
\usepackage{amsmath}
\usepackage{amsthm}
\usepackage{amssymb}
\usepackage[numbers]{natbib}
\usepackage[unicode=true,pdfusetitle,
 bookmarks=true,bookmarksnumbered=true,bookmarksopen=false,
 breaklinks=false,pdfborder={0 0 1},backref=false,colorlinks=true]
 {hyperref}
\hypersetup{
 pdfborderstyle=,linkcolor=black,citecolor=blue,urlcolor=black,filecolor=blue,pdfpagelayout=OneColumn,pdfnewwindow=true,pdfstartview=XYZ,plainpages=false}

\makeatletter
\theoremstyle{plain}
\newtheorem{thm}{\protect\theoremname}
\theoremstyle{remark}
\newtheorem{rem}[thm]{\protect\remarkname}

\@ifundefined{date}{}{\date{}}
\usepackage{babel}
\usepackage{caption}
\usepackage[nottoc]{tocbibind}
\allowdisplaybreaks[4]
\usepackage{dsfont}
\DeclareMathAlphabet{\mathcal}{OMS}{cmsy}{m}{n}

\providecommand{\remarkname}{Remark}
\providecommand{\theoremname}{Theorem}

\makeatother

\providecommand{\remarkname}{Remark}
\providecommand{\theoremname}{Theorem}

\begin{document}
\global\long\def\Sgm{\boldsymbol{\Sigma}}%
\global\long\def\W{\boldsymbol{W}}%
\global\long\def\H{\boldsymbol{H}}%
\global\long\def\P{\mathbb{P}}%
\global\long\def\Q{\mathbb{Q}}%
\global\long\def\xx{\bm{x}}%
\global\long\def\dd{\mathrm{d}}%

\title{On the transport equation for probability density functions of turbulent
vorticity fields}
\author{Jiawei Li\thanks{Department of Mathematical Sciences, Carnegie Mellon University, Pittsburgh,
PA 15213, United States. \emph{Email}:\protect\protect\protect\protect\protect\protect\protect\href{mailto:jiaweil4@andrew.cmu.edu}{jiaweil4@andrew.cmu.edu}},$\;\;$Zhongmin Qian\thanks{Mathematical Institute, University of Oxford, Oxford OX2 6GG, United
Kingdom, and Oxford Suzhou Center for Advanced Research. \emph{Email}:
\protect\protect\protect\protect\protect\protect\protect\href{mailto:zhongmin.qian@maths.ox.ac.uk}{zhongmin.qian@maths.ox.ac.uk}} \ and\ Mingrui Zhou\thanks{Mathematical Institute, University of Oxford, Oxford OX2 6GG, United
Kingdom. \emph{Email}: \protect\protect\protect\protect\protect\protect\protect\href{mailto:mingrui.zhou@maths.ox.ac.uk}{mingrui.zhou@maths.ox.ac.uk}}}
\maketitle
\begin{abstract}
Vorticity random fields of turbulent flows (modeled over the vorticity
equation with random initial data for example) are singled out as
the main dynamic variables for the description of turbulence, and
the evolution equation of the probability density function (PDF) of
the vorticity field has been obtained. This PDF evolution equation
is a mixed type partial differential equation (PDE) of second order
which depends only on the conditional mean (which is a first order
statistics) of the underlying turbulent flow. This is in contrast
with Reynolds' mean flow equation which relies on a quadratic statistics.
The PDF PDE may provide new closure schemes based on the first order
conditional statistics, and some of them will be described in the
paper. We should mention that the PDF equation is interesting by its
own and is worthy of study as a PDE of second order.

\medskip{}


\medskip{}

\end{abstract}

\section{Introduction}

In statistical fluid mechanics (cf. \citep{MoninandYaglom1965}),
the velocity $U(x,t)$ of a turbulent flow is promoted to a random
field, cf. \citep{K41-a}, indexed by space variable $x\in\mathbb{R}^{3}$
and time parameter $t$. From this point of view, the turbulence problem,
if there is one, seeks for a description of the distribution of the
velocity field. This distribution is rather complicated and consists
of all joint distributions of the velocity across over finitely many
locations and times, and therefore it is challenging to describe the
distributions of turbulent flows in general. As early as in 1950's,
E. Hopf \citep{Hopf1952} (cf. \citep{MoninandYaglom1965} too) made
an ambitious attempt and derived differential equations for the distribution
of a turbulent flow. Hopf's differential equations are however infinite
dimensional and involve functional derivatives, which are therefore
too difficult to extract useful information about turbulence. The
dominant approach in the statistical theory of turbulence, initiated
in G. I. Taylor's seminal work \citep{Taylor1921,Taylor1935}, has
been based on the analysis of moment structure functions and relied
on the spectral method (cf. \citep{Batchelor1953,MoninandYaglom1965}
for example). In the past decades, attention has been paid to the
one-dimensional marginal distribution of the velocity field $U(x,t)$,
whose distribution is finite-dimensional. It is reasonable to assume
that the distribution of $U(x,t)$, with $(x,t)$ fixed, has a probability
density function (PDF) with respect to the Lebesgue measure on $\mathbb{R}^{3}$.
The PDF methods based on formal PDF transport equations have been
developed through a series of work by Pope and other researchers (cf.
\citep{Pope1985,Pope2000}), which become powerful tools for modeling
turbulent flows. The PDF, by definition, has to satisfy ``the adjoint
equation'' of the fluid dynamic equations. If the turbulent flow
is an incompressible viscous fluid flow, then the PDF of the velocity
must satisfy the adjoint equation of the Navier-Stokes equation. The
transport equation for PDF known in the literature is still a formal
adjoint equation of the Navier-Stokes equation, and therefore only
few features from the transport equation may be used in modeling turbulent
flows.

For incompressible fluid flows, the Navier-Stokes equation is equivalent
to the vorticity equation (see (\ref{eq:vort-tu1}) below), and therefore
it is natural to consider the vorticity $W=\nabla\wedge U$ as the
main fluid dynamic variable for the study of turbulence \citep{Moffatt2011}.
There are good reasons why we should concentrate on the vortex motion
in turbulence. The velocity of a turbulent flow is unlikely to be
independently or possess conditional independence with respect to
the spatial variable, while vortex motions of many turbulent flows
observed in nature (such as vortex lines, vortex rings) acquire certain
conditional independence, in the sense that by focusing on the motion
near a fixed region, the future vortex motions evolve more or less
independent of what happens in other positions. In fact there are
good evidences which demonstrate that some sort of superposition property
of vorticity may be maintained, although not exactly due to highly
non-linear and non-local nature of turbulence. These observations
are valuable in modeling turbulent flows via the vorticity, which
are already applied in vortex methods (cf. \citep{CottetandKoumoutsakos2000,MajadaBertozzi2002}).
In this sense PDF methods based on the vorticity are valuable.

The main contribution of the present work is the partial differential
equation for the PDF of the vorticity (called PDF PDE, or PDF equation
for short), which will be derived in the main body of the paper. The
PDF PDE is highly non-linear, however the most striking aspect is
that the PDF PDE for the vorticity depends only on one single first-order
statistical characteristic of the turbulent flow. More precisely,
we identify the main statistical characteristic needed for the PDF
PDE with the conditional mean function of the increment of the vorticity
given the current vorticity: 
\begin{equation}
\mu^{i}(x,y,w,t)=\mathbb{E}\left[\left.W^{i}(y,t)-W^{i}(x,t)\right|W(x,t)=w\right]\label{eq:mu-average}
\end{equation}
for $i=1,2,3$, where $x,y,w\in\mathbb{R}^{3}$ and $t\geq0$. The
PDF of the vorticity $W$ is a solution to the PDF PDE, which is a
second-order partial differential equation where coefficients appearing
in the PDF PDE depend on $\mu$ only. Besides its theoretical interest,
this PDF PDE paves the way towards practical modeling the statistics
of turbulent flows based on the vorticity PDF.

The paper is organized as the following. In Section 2, several notions
and notations together with several standard assumptions about fluid
dynamic random fields will be introduced, and the main result, namely
the PDF PDE for vorticity will be derived. In Sections 3 and 4, the
theoretical foundation, based on our PDF PDE for vorticity, will be
laid for the purpose of modeling various turbulent flows. In Section
5, we propose the most direct way of modeling PDF of the vorticity
by specifying the conditional mean function $\mu^{i}$. The PDF PDE
considered purely as a PDE theory is over-determined, due to its non-linearlity
in the sense that the coefficients appearing in the PDF PDE are not
independent of its solutions. Therefore care is needed to ensure that
the additional constraint is satisfied. In the last section, we propose
the heat flow method to model statistical quantities needed for closing
the PDF PDE and obtain concrete PDF examples.

\emph{Conventions on notations}. The following conventions are used
throughout the paper. Firstly Einstein's convention on summation on
repeated indices through their ranges is assumed, unless otherwise
specified. If $A$ is a vector or a vector field (in the space of
dimension three) dependent on some parameters, then its components
are labeled with upper-script indices, i.e. $A=(A^{i})=(A^{1},A^{2},A^{3})$.
The same convention applies to coordinates too. Partial derivatives
of functions may be labeled with variables in sub-scripts. For example,
if $A(w;x,t)$ is a vector-valued function depending on $w,x\in\mathbb{R}^{3}$
and $t$, then $\nabla_{w}A$ means the total derivatives $(\frac{\partial}{\partial w^{i}}A^{j})$,
$\Delta_{x}A$ means the vector $(\frac{\partial^{2}}{\partial x^{k}\partial x^{k}}A^{i})$
of the Laplacians of $A^{i}$. However, as a general rule, derivatives
without sub-scripts mean the derivatives with respect to the variable
$x=(x^{i})$, unless otherwise specified for avoiding possible confusion.

The velocity vector field will be denoted by $U$. $W=\nabla\wedge U$
is its vorticity so that $W^{i}=\varepsilon^{ijk}\frac{\partial}{\partial x^{j}}U^{k}$,
where $\varepsilon^{ijk}$ are the Levi-Civita symbols.

\section{PDF equation for the vorticity}

In this paper we regard the vorticity of an incompressible turbulent
flow as the main dynamic variable. The goal of this section is to
derive the evolution equation for the PDF of the vorticity of an incompressible
turbulent flow. 

\subsection{\label{Sec2.1}Prelims and assumptions}

Let $U(x,t)$ be the velocity of an incompressible turbulent flow
in $\mathbb{R}^{3}$ with viscosity $\nu$. Suppose there is no external
force supplied to the turbulence. Being a random field though, $U(x,t)$
satisfies the Navier-Stokes equations 
\[
\frac{\partial U^{i}}{\partial t}+U^{j}\frac{\partial U^{i}}{\partial x^{j}}=\nu\Delta U^{i}-\frac{\partial P}{\partial x^{i}}
\]
and 
\[
\frac{\partial U^{j}}{\partial x^{j}}=0,
\]
where $i=1,2,3$, and $P$ is the pressure, subject to initial conditions
which are random. The motion equations of the vorticity $W=\nabla\wedge U$
are the vorticity equations 
\begin{equation}
\frac{\partial W^{i}}{\partial t}+U^{j}\frac{\partial W^{i}}{\partial x^{j}}=\nu\Delta W^{i}+W^{j}\frac{\partial U^{i}}{\partial x^{j}}\label{eq:vort-tu1}
\end{equation}
for $i=1,2,3$.

We make the following technical assumptions.

First, we assume that both $U(x,t)$ and $W(x,t)$ have derivatives
in $x$ and $t$ of any order, and theses derivatives decay to zero
at infinity sufficiently fast, so that possible boundary terms arising
in applications of the Stokes' formula have no contributions in computations
below. Therefore, since $U$ is divergence-free, $\Delta U=-\nabla\wedge W$,
according to Green's formula
\[
U^{i}(x,t)=\int_{\mathbb{R}^{3}}\frac{1}{4\pi|x-y|}\varepsilon^{ijk}\frac{\partial}{\partial y^{j}}W^{k}(y,t)\mathrm{d}y,
\]
which yields the Biot-Savart law: 
\begin{equation}
U^{i}(x,t)=-\int_{\mathbb{R}^{3}}\varepsilon^{ijk}\frac{x^{j}-y^{j}}{4\pi|x-y|^{3}}W^{k}(y,t)\dd y.\label{eq:b-s form1}
\end{equation}
for $i=1,2,3$.

Our second technical assumption is to impose certain regularity on
the distribution of the vorticity. At each point $x$ and instance
$t\geq0$, $W(x,t)$ is a random variable defined on a probability
space $(\Omega,\mathcal{F},\mathbb{P})$, taking values in $\mathbb{R}^{3}$,
and therefore its law (or called distribution) is a probability measure
on $(\mathbb{R}^{3},\mathcal{B}(\mathbb{R}^{3}))$. We assume that
$W(x,t)$ has a positive and smooth probability density function (PDF),
denoted by $f(w;x,t)$, in the sense that
\[
\mathbb{P}\left(W(x,t)\in A\right)=\int_{A}f(w;x,t)\dd w
\]
for every Borel measurable subset $A$. The joint distribution of
$W(x,t)$ and $W(y,t)$ for any pair $x\neq y$ has a positive and
smooth PDF, denoted by $f_{2}(w_{1},w_{2};x,y,t)$. 

To state our third assumption, we first notice that 
\[
f_{2|1}(w;w_{1},x,y,t)=\frac{f_{2}(w_{1},w;x,y,t)}{f(w_{1};x,t)}
\]
is the PDF of the conditional law of $W(y,t)$ given $W(x,t)=w_{1}$.
The conditional mean function $\mu^{i}(x,y,w,t)$ is defined by 
\begin{align}
\mu^{i}(x,y,w,t) & =\mathbb{E}\left[\left.W^{i}(y,t)-W^{i}(x,t)\right|W(x,t)=w\right]\nonumber \\
 & =\int_{\mathbb{R}^{3}}(w_{1}^{i}-w^{i})f_{2|1}(w_{1};w,x,y,t)\dd w_{1}\label{eq:cd-mean}
\end{align}
for $i=1,2,3$, which will play a dominant role in the sequel. 

The third technical assumption is about the regularity of the conditional
average function $\mu=(\mu^{1},\mu^{2},\mu^{3})$. It is assumed that
the derivatives of $\mu$ of any order exist and decay to zero sufficiently
fast at infinity. Moreover it is assumed that $\mu$ has an asymptotic
expansion 
\begin{equation}
\mu^{i}(x,y,w,t)=a_{k}^{i}(x,w,t)(y^{k}-x^{k})+b_{jk}^{i}(x,w,t)(y^{k}-x^{k})(y^{j}-x^{j})+o\left(|y-x|^{2}\right)\label{eq:mean asy}
\end{equation}
as $|y-x|\rightarrow0$, where $a_{k}^{i}$'s and $b_{jk}^{i}$'s
are assumed to be continuous with respect to all of their arguments.
We denote $b^{i}=b_{kk}^{i}$ for $i=1,2,3$. It is clear that 
\begin{equation}
a_{k}^{i}(x,w,t)=\mathbb{E}\left[\left.\frac{\partial W^{i}}{\partial x^{k}}(x,t)\right|W(x,t)=w\right]=\left.\frac{\partial}{\partial y^{k}}\mu^{i}(x,y,w,t)\right|_{y=x}\label{eq:a_k defn}
\end{equation}
which represents the local rate of change in the vortex motion over
the turbulent region, and 
\begin{equation}
b^{i}(x,w,t)=\frac{1}{2}\mathbb{E}\left[\Delta_{x}W^{i}(x,t)|W(x,t)=w\right]=\frac{1}{2}\left.\Delta_{y}\mu^{i}(x,y,w,t)\right|_{y=x}.\label{eq:b defn}
\end{equation}

In the remainder of the paper, we will work with a turbulent flow
for which the three assumptions listed above are satisfied. 

\subsection{PDF equation and its derivation}

In this section, we derive the main result of the paper, that is,
a partial differential equation which the PDF of the vorticity must
satisfy. 
\begin{thm}
\label{thm:pdf-v}Under the assumptions and notations established
in Sec. \ref{Sec2.1}, suppose the PDF $f(w;x,t)$ of the vorticity
is smooth in $(w,x,t)$ and has finite moments, that is $\int_{\mathbb{R}^{3}}|w|^{n}f(w;x,t)dw<\infty$
for every $n=1,2,\cdots$. Then $f$ satisfies the following partial
differential equation: 
\begin{equation}
\left(\frac{\partial}{\partial t}+\frac{\partial B^{i}}{\partial x^{i}}+B^{i}\frac{\partial}{\partial x^{i}}-\nu\Delta_{x}\right)f=\nu\frac{\partial}{\partial w^{i}}\left(\frac{\partial}{\partial x^{k}}\left(fa_{k}^{i}\right)-2b^{i}f\right)+\frac{\partial}{\partial w^{i}}\left(fD^{i}\right),\label{eq:vot-pdf1}
\end{equation}
where $\Delta_{x}$ is the Laplacian with respect to the space variable
$x$, $a_{k}$ and $b$ are given as in (\ref{eq:a_k defn}) and (\ref{eq:b defn}),
\begin{equation}
B^{i}(x,w,t)=\int_{\mathbb{R}^{3}}\frac{1}{4\pi|y-x|}\varepsilon^{ijk}\frac{\partial}{\partial y^{j}}\mu^{k}(x,y,w,t)\dd y,\label{eq:B-vector}
\end{equation}
and 
\begin{equation}
D^{i}(x,w,t)=w^{l}\int_{\mathbb{R}^{3}}\frac{y^{l}-x^{l}}{4\pi|y-x|^{3}}\varepsilon^{ijk}\frac{\partial}{\partial y^{j}}\mu^{k}(x,y,w,t)\dd y.\label{eq:D-vector}
\end{equation}
\end{thm}

\begin{proof}
Let $F$ be a smooth function on $\mathbb{R}^{3}$ with a compact
support. By the definition of PDF one has
\[
\int_{\mathbb{R}^{3}}F(w)\frac{\partial}{\partial t}f(w;x,t)\dd w=\mathbb{E}\left[\frac{\partial}{\partial t}F(W(x,t))\right].
\]
We are going to calculate the right-hand side expectation in terms
of the PDF $f(w;x,t)$ and other statistical characteristics of the
turbulent flow. This will be done by using the following equation
\[
\frac{\partial}{\partial t}F(W(x,t))=\nu\Delta_{x}F(W)-\frac{\partial\left(U^{i}F(W)\right)}{\partial x^{i}}+F_{j}(W)W^{i}\frac{\partial U^{j}}{\partial x^{i}}-\nu\frac{\partial F_{j}(W)}{\partial x^{k}}\frac{\partial W^{j}}{\partial x^{k}},
\]
where, for simplicity, $F_{j}$ denote the partial derivatives $\frac{\partial F}{\partial x^{j}}$
of $F$, $j=1,2,3$. This equation follows directly from the vorticity
equations. From the previous equation we obtain that 
\begin{align}
\mathbb{E}\left[\frac{\partial}{\partial t}F(W(x,t))\right] & =\nu\mathbb{E}\left[\Delta_{x}F(W)\right]-\mathbb{E}\left[\frac{\partial\left(U^{i}F(W)\right)}{\partial x^{i}}\right]\nonumber \\
 & +\mathbb{E}\left[F_{j}(W)W^{i}\frac{\partial U^{j}}{\partial x^{i}}\right]-\mathbb{E}\left[\nu\frac{\partial F_{j}(W)}{\partial x^{k}}\frac{\partial W^{j}}{\partial x^{k}}\right]\nonumber \\
 & :=I_{1}+I_{2}+I_{3}+I_{4}.\label{eq:deriv eqn sum}
\end{align}
Let us calculate $I_{i}$'s on the right-hand side of Eq. (\ref{eq:deriv eqn sum}).
First, it is easy to see that by definition 
\begin{align}
I_{1} & =\nu\Delta_{x}\mathbb{E}\left[F(W)\right]=\int_{\mathbb{R}^{3}}F(w)\nu\Delta_{x}f(w;x,t)\dd w.\label{I-1}
\end{align}
For computing $I_{2}$, we shall use Eq. (\ref{eq:b-s form1}) and
obtain that 
\begin{align*}
I_{2} & =-\mathbb{E}\left[\frac{\partial(U^{i}F(W))}{\partial x^{i}}\right]=-\frac{\partial}{\partial x^{i}}\mathbb{E}\left[F(W)U^{i}\right]\\
 & =-\frac{\partial}{\partial x^{i}}\mathbb{E}\left[F(W)\varepsilon^{ijk}\int_{\mathbb{R}^{3}}\frac{1}{4\pi|y-x|}\frac{\partial}{\partial y^{j}}W^{k}(y,t)\dd y\right].
\end{align*}
To work out the expectation on the right-hand side, we may rewrite
the partial derivative as a limit: 
\[
\frac{\partial}{\partial y^{j}}W^{k}(y,t)=\lim_{h\rightarrow0}\frac{1}{h}\left(W^{k}(y+he^{(j)},t)-W^{k}(y,t)\right),
\]
where $e^{(j)}$ represents the unit vector with $j$-th component
equal to $1$, and the rest two components $0$, so that we can rewrite
$I_{2}$ as the following limit
\begin{equation}
I_{2}=-\frac{\partial}{\partial x^{i}}\lim_{h\rightarrow0}\frac{1}{h}\mathbb{E}\left[F(W)\varepsilon^{ijk}\int_{\mathbb{R}^{3}}\frac{1}{4\pi|y-x|}\left(W^{k}(y+he^{(j)},t)-W^{k}(y,t)\right)\dd y\right].\label{I-2-m}
\end{equation}
The expectation in this expression may be written in terms of two-point
joint distributions of $W$ as the following 
\[
\int_{\mathbb{R}^{3}}F(w)\left[\int\frac{\varepsilon^{ijk}}{4\pi|y-x|}\left(\int_{\mathbb{R}^{3}}w_{1}^{k}\left(f_{2}(w,w_{1};x,y+he^{(j)},t)-f_{2}(w,w_{1};x,y,t)\right)\dd w_{1}\right)\dd y\right]\dd w.
\]
The inner integral against the variable $w_{1}$ equals 
\[
f(w;x,t)\left(\mu(x,y+he^{(j)},w,t)-\mu(x,y,w,t)\right).
\]
After substituting this in Eq. \eqref{I-2-m} and sending $h\rightarrow0$,
we obtain that 
\begin{equation}
I_{2}=\int_{\mathbb{R}^{3}}F(w)\left[-\frac{\partial}{\partial x^{i}}\left(f(w;x,t)\int_{\mathbb{R}^{3}}\frac{\varepsilon^{ijk}}{4\pi|y-x|}\frac{\partial}{\partial y^{j}}\mu^{k}(x,y,w,t)\dd y\right)\right]\dd w.\label{I-2}
\end{equation}
Next we deal with $I_{3}$. Again, using (\ref{eq:b-s form1}), we
may write 
\[
\frac{\partial}{\partial x^{j}}U^{i}(x,t)=\int_{\mathbb{R}^{3}}\frac{y^{j}-x^{j}}{4\pi|y-x|^{3}}\varepsilon^{ilk}\frac{\partial}{\partial y^{l}}W^{k}(y,t)\dd y,
\]
so that 
\begin{align*}
I_{3} & =\mathbb{E}\left[F_{i}(W)W^{j}\frac{\partial U^{i}}{\partial x^{j}}\right]\\
 & =\mathbb{E}\left[F_{i}(W)W^{j}\int_{\mathbb{R}^{3}}\frac{y^{j}-x^{j}}{4\pi|y-x|^{3}}\varepsilon^{ilk}\frac{\partial}{\partial y^{l}}W^{k}(y,t)\dd y\right],
\end{align*}
which can be evaluated by using the two-point joint PDF. Indeed, we
may repeat the same idea as in the computation of $I_{2}$, to obtain
that
\begin{equation}
I_{3}=\int_{\mathbb{R}^{3}}F(w)\frac{\partial}{\partial w^{i}}\left[-f(w;x,t)\varepsilon^{ilk}w^{j}\int_{\mathbb{R}^{3}}\frac{y^{j}-x^{j}}{4\pi|y-x|^{3}}\frac{\partial}{\partial y^{l}}\mu^{k}(x,y,w,t)\dd y\right]\dd w.\label{I-3}
\end{equation}
Similarly, for the last term $I_{4}$, we write 
\begin{align*}
I_{4} & =-\nu\sum_{k=1}^{3}\mathbb{E}\left[\frac{\partial F_{i}(W)}{\partial x^{k}}\frac{\partial W^{i}}{\partial x^{k}}\right]\\
 & =-\nu\sum_{k=1}^{3}\lim_{h\rightarrow0}\frac{1}{h^{2}}\mathbb{E}\left[\left(F_{i}(W(x+he^{(k)},t))-F_{i}(W(x,t))\right)\left(W^{i}(x+he^{(k)},t)-W^{i}(x,t)\right)\right]\\
 & =-\nu\sum_{k=1}^{3}\lim_{h\rightarrow0}\frac{1}{h^{2}}\iint_{\mathbb{R}^{3}\times\mathbb{R}^{3}}\left(F_{i}(w_{2})-F_{i}(w_{1})\right)\left(w_{2}^{i}-w_{1}^{i}\right)f_{2}(w_{1},w_{2},x,x+he^{(k)},t)\dd w_{1}\dd w_{2}\\
 & =\nu\int_{\mathbb{R}^{3}}F_{i}(w)\lim_{h\rightarrow0}\frac{1}{h^{2}}\sum_{k=1}^{3}\left\{ f(w;x+he^{(k)},t)\int_{\mathbb{R}^{3}}\left(w_{1}^{i}-w^{i}\right)f_{2|1}(w_{1};w,x+he^{(k)},x,t)\dd w_{1}\right.\\
 & +\left.f(w;x,t)\int_{\mathbb{R}^{3}}\left(w_{1}^{i}-w^{i}\right)f_{2|1}(w_{1};w,x,x+he^{(k)},t)\dd w_{1}\right\} \dd w\\
 & =\nu\int_{\mathbb{R}^{3}}F_{i}(w)\lim_{h\rightarrow0}\frac{1}{h^{2}}I_{4}^{h}dw,
\end{align*}
where for simplicity we have introduced the notation
\[
\begin{aligned}I_{4}^{h}:= & \sum_{k=1}^{3}\left\{ \left(f(w;x+he^{(k)},t)-f(w;x,t)\right)\mu^{i}(x+he^{(k)},x,w,t)\right.\\
 & +\left.f(w;x,t)\left(\mu^{i}(x,x+he^{(k)},w,t)+\mu^{i}(x+he^{(k)},x,w,t)\right)\right\} .
\end{aligned}
\]
$I_{4}^{h}$ can be evaluated by using (\ref{eq:mean asy}) when $h$
is sufficiently small, by which we mean that 
\[
\mu^{i}(x,x+he^{(k)},w,t)=a_{k}^{i}(x,w,t)h+b_{kk}^{i}(x,w,t)h^{2}+o\left(h^{2}\right)
\]
and 
\[
\mu^{i}(x+he^{(k)},x,w,t)=-a_{k}^{i}(x+he^{(k)},w,t)h+b_{kk}^{i}(x+he^{(k)},w,t)h^{2}+o\left(h^{2}\right).
\]
Consequently, we have 
\begin{align*}
\mu^{i}(x,x+he^{(k)},w,t)+\mu^{i}(x+he^{(k)},x,w,t) & =-\left(a_{k}^{i}(x+he^{(k)},w,t)-a_{k}^{i}(x,w,t)\right)h\\
 & +\left(b_{kk}^{i}(x,w,t)+b_{kk}^{i}(x+he^{(k)},w,t)\right)h^{2}+o\left(h^{2}\right),
\end{align*}
which yields that 
\[
\lim_{h\rightarrow0}\frac{1}{h^{2}}I_{4}^{h}=-\frac{\partial}{\partial x^{k}}\left(fa_{k}^{i}\right)+2f\sum_{k}b_{kk}^{i}.
\]
and we may conclude that 
\begin{align}
I_{4} & =-\nu\int_{\mathbb{R}^{3}}F_{i}(w)\left[\frac{\partial}{\partial x^{k}}\left(fa_{k}^{i}\right)-2f\sum_{k}b_{kk}^{i}\right]\dd w\nonumber \\
 & =\int_{\mathbb{R}^{3}}F(w)\frac{\partial}{\partial w^{i}}\left[\nu\frac{\partial}{\partial x^{k}}\left(fa_{k}^{i}\right)-2\nu f\sum_{k}b_{kk}^{i}\right]\dd w.\label{I-4}
\end{align}
The PDF equation \eqref{eq:vot-pdf1} now follows by substituting
(\ref{I-1}, \ref{I-2}, \ref{I-3}, \ref{I-4}) into (\ref{eq:deriv eqn sum}).
\end{proof}
\begin{rem}
The partial differential equation (\ref{eq:vot-pdf1}), although it
appears linear, is in fact highly non-linear, and more importantly,
the coefficients $a,b,B$ and $D$ that define the partial differential
equation cannot be in general determined by the PDF $f$ alone. We
have seen that these functions are functionals of the conditional
average function $\mu$, hence the PDF $f$ cannot be obtained solely
through the PDF equation. However, the PDF PDE provides an appealing
and new method for modeling the PDF $f$ based on the modeling of
$\mu$ alone, we will provide in this work several results on modeling
the PDF $f$ based on the PDF equation we just derived. 
\end{rem}

\begin{rem}
If all coefficients $a,b,B$ and $D$ are considered as given, then
we can pose the initial value problem for solving the PDF PDE \eqref{eq:vot-pdf1}.
The good news is that then the PDF PDE is a scalar linear partial
differential equation of second order, although in general it is a
mixed type of partial differential equations of second order on 6
dimensional space. To the best knowledge of the authors of the present
paper, this kind of partial differential equations has not been studied
systemically in the existing literature. 
\end{rem}

\subsection{The distribution of turbulence}

As an application of the PDF PDE obtained in the previous sub-section,
we address a long-standing question in turbulence about the distribution
of turbulent flows. It has been conjectured and verified by measurements
over many years (cf. \citep{Batchelor1953} and \citep{Davidson2004})
that the distribution of a genuine turbulent flow cannot be Gaussian,
while a mathematical proof for this statement, to the best knowledge
of the present authors, is not available yet. By using the PDF PDE
we are able to prove the following
\begin{thm}
Consider an incompressible viscous (with viscosity $\nu>0$) turbulent
flow with vorticity $W(x,t)$. Suppose

1) the mean vorticity is constant at any instance,

2) $\{W(x,t)\}$ is weakly isotropic in the sense that 
\[
\mathbb{E}\left[\nabla W(x,t)|W(x,t)=w\right]=0
\]
for all $x,w$ and $t\geq0$, and

3) the distribution of $\{W(x,t)\}$ is Gaussian.

Then the distribution of $W(x,t)$ is independent of $x$. 
\end{thm}

\begin{proof}
Since $\left\{ W(x,t):x\in\mathbb{R}^{d}\textrm{ and }t\geq0\right\} $
is a Gaussian random field, for $x\neq y$, the joint law of $W(x,t)$
and $W(y,t)$ is a normal distribution on $\mathbb{R}^{6}$ whose
co-variance matrix $\Sigma(x,y)$ may be decomposed into blocks 
\[
\Theta(x,y)=\left(\begin{array}{cc}
\Sigma_{x} & \Sigma_{x,y}\\
\Sigma_{y,x} & \Sigma_{y}
\end{array}\right),
\]
where the dependence on $t$ is suppressed for simplicity. The PDF
of $W(x,t)$ is given by
\[
f(w;x,t)=\frac{1}{(2\pi)^{3/2}\sqrt{\det\Sigma_{x}}}\exp\left(-\frac{1}{2}(w-m)^{T}\Sigma_{x}^{-1}(w-m)\right),
\]
where $m$ denotes the mean vorticity, so that
\begin{equation}
\ln f(w;x,t)=-\frac{3}{2}\ln(2\pi)-\frac{1}{2}\ln\det\Sigma_{x}-\frac{1}{2}(w-m)^{T}\Sigma_{x}^{-1}(w-m).\label{phi-1}
\end{equation}
Since under the assumption $a_{k}^{i}=0$, $\psi=\ln f$ according
to PDF PDE \eqref{eq:vot-pdf1} must satisfy the following equation:
\begin{align*}
\frac{\partial\psi}{\partial t} & =\nu\Delta_{x}\psi-B^{i}\frac{\partial\psi}{\partial x^{i}}+\nu|\nabla_{x}\psi|^{2}\\
 & +\left(D^{i}-2\nu b^{i}\right)\frac{\partial\psi}{\partial w^{i}}+\frac{\partial\left(D^{i}-2\nu b^{i}\right)}{\partial w^{i}}-\nabla_{x}\cdot B
\end{align*}
(which is the PDF PDE for $\psi=\ln f$ in the case where $a_{k}^{i}$
vanish). Using \eqref{phi-1} and by a lengthy but completely elementary
computation, we deduce that
\begin{align*}
-2\frac{\partial\psi}{\partial t} & =\nu\Delta_{x}\ln\det\Sigma_{x}-\frac{1}{2}\nu\left|\nabla_{x}\ln\det\Sigma_{x}\right|^{2}+4\nu\frac{\partial b^{i}}{\partial w^{i}}\\
 & -B^{i}\frac{\partial\ln\det\Sigma_{x}}{\partial x^{i}}-2\frac{\partial D^{i}}{\partial w^{i}}+2\nabla_{x}\cdot B+\nu u^{T}\Delta_{x}\Sigma_{x}^{-1}u\\
 & -\nu u^{T}\left(\sum_{i}\frac{\partial\ln\det\Sigma_{x}}{\partial x^{i}}\frac{\partial\Sigma_{x}^{-1}}{\partial x^{i}}\right)u-4\nu b^{i}(\Sigma_{x}^{-1})_{i\beta}u^{\beta}\\
 & +2D^{i}(\Sigma_{x}^{-1})_{i\beta}u^{\beta}-B^{i}u^{T}\frac{\partial\Sigma_{x}^{-1}}{\partial x^{i}}u-\frac{1}{2}\nu\left|u^{T}\nabla_{x}\Sigma_{x}^{-1}u\right|^{2},
\end{align*}
where $u=w-m$. The last equation looks complicated, but the important
point in our argument is the observation that, the left-hand side
$-2\frac{\partial\psi}{\partial t}$ is a quadratic function of $u$,
while the right-hand side is a polynomial in $u$ of degree $4$,
and only the last term on the right-hand side has order 4 in $u$.
Therefore 
\[
\frac{1}{2}\nu\left|u^{T}\nabla_{x}\Sigma_{x}^{-1}u\right|^{2}=0
\]
for every $u$, which yields that $\nabla_{x}\Sigma_{x}^{-1}=0$,
so that $\Sigma_{x}$ is independent of $x$. This completes the proof.
\end{proof}

\section{Inviscid fluid flows}

There is a significant simplification in PDF PDE for an inviscid ``turbulent''
flow, although such turbulence may not exist in nature. For inviscid
fluid flows, the velocity $U(x,t)$ satisfies the Euler equations
\[
\frac{\partial}{\partial t}U^{i}+U^{j}\frac{\partial U^{i}}{\partial x^{j}}=-\nabla P,\quad\frac{\partial U^{i}}{\partial x^{i}}=0.
\]
The PDF $f(w;x,t)$ of $W=\nabla\wedge U$ satisfies the (non-linear)
transport differential equation 
\begin{equation}
\frac{\partial f}{\partial t}+\frac{\partial}{\partial x^{i}}\left(fB^{i}\right)=\frac{\partial}{\partial w^{i}}\left(fD^{i}\right),\label{eq: inviscid1}
\end{equation}
where $B$ and $D$ are given as in (\ref{eq:B-vector}) and (\ref{eq:D-vector})
respectively.

The following theorem provides a mathematical tool for modeling PDF
of an inviscid turbulent flow. 
\begin{thm}
\label{thm:Invicid case} Consider PDE \eqref{eq: inviscid1} where
the data $B$ and $D$ are assumed as given. Assume that $B$ and
$D$ are Lipschitz continuous in $(x,w)$ uniformly in $t>0$. Suppose
$f(w;x,t)$ is a solution to \eqref{eq: inviscid1} with continuous
initial data $f(w;x,0)$, and suppose $fD$ decays to zero sufficiently
fast as $|w|\to\infty$. Then we have the followings:

1) $f(w;x,0)\geq0$ for all $x\in\mathbb{R}^{3}$, then $f(w;x;t)\geq0$
for all $x\in\mathbb{R}^{3}$ and $t>0$.

2) Suppose $\int_{\mathbb{R}^{3}}f(w;x,0)\mathrm{d}w=1$ for all $x\in\mathbb{R}^{3}$,
then $\int_{\mathbb{R}^{3}}f(w;x,t)\mathrm{d}w=1$ for all $x\in\mathbb{R}^{3}$
and $t>0$, if and only if 
\begin{equation}
\frac{\partial}{\partial x^{i}}\int_{\mathbb{R}^{3}}B^{i}(x,w,t)f(w;x,t)\mathrm{d}w=0\label{eq:inviscid condition}
\end{equation}
for all $x\in\mathbb{R}^{3}$ and $t>0$. 
\end{thm}

\begin{proof}
Let us define the integral curves $X$ and $Y$ by the following ordinary
differential equation system: 
\[
\begin{cases}
\frac{\dd}{\dd s}X^{i}(t,s)=-B^{i}(X(t,s),Y(t,s),t-s)), & X(t,0)=x,\\
\frac{\dd}{\dd s}Y^{i}(t,s)=D^{i}(X(t,s),Y(t,s),t-s)), & Y(t,0)=w.
\end{cases}
\]
Define
\[
h(s)=f(Y(t,s);X(t,s),t-s)e^{N(w;x,s)}
\]
for all $s\in[0,t]$, where 
\[
N(w;x,s):=\int_{0}^{s}\left(\frac{\partial D^{i}}{\partial w^{i}}-\frac{\partial B^{i}}{\partial x^{i}}\right)(X(t,r);Y(t,r),t-r)\dd r.
\]
Then clearly $h(0)=f(w;x,t)$ and 
\[
h(t)=f(Y(t,t);X(t,t),0)e^{N(w;x,t)}.
\]
It is clear that 
\begin{align*}
\frac{\dd}{\dd s}h(s)= & e^{N}\frac{\partial f}{\partial w^{i}}\frac{\dd}{\dd s}Y^{i}(t,s)+e^{N}\frac{\partial f}{\partial x^{i}}\frac{\dd}{\dd s}X^{i}(t,s)\\
 & -e^{N}\frac{\partial}{\partial t}f+e^{N}f\frac{\partial D^{i}}{\partial w^{i}}-e^{N}f\frac{\partial B^{i}}{\partial x^{i}}\\
= & e^{N}\left\{ \frac{\partial(fD^{i})}{\partial w^{i}}-\frac{\partial(B^{i}f)}{\partial x^{i}}-\frac{\partial}{\partial t}f\right\} \\
= & 0.
\end{align*}
Integrating with respect to the variable $s$ over $[0,t]$, we may
conclude that 
\[
f(w;x,t)=f(Y(t,t);X(t,t),0)\exp\left[\int_{0}^{t}\left(\frac{\partial D^{i}}{\partial w^{i}}-\frac{\partial B^{i}}{\partial x^{i}}\right)(X(t,s);Y(t,s),t-s)\dd s\right].
\]

The conservation of positivity property 1) follows immediately from
this representation. To show 2) we consider the function $u(x,t)=\int_{\mathbb{R}^{3}}f(w;x,t)\dd w$
for $x\in\mathbb{R}^{3}$ and $t\geq0$. If $u(x,t)=1$ for all $x$
and $t>0$, then by integrating \eqref{eq: inviscid1} over $\mathbb{R}^{3}$
we obtain \eqref{eq:inviscid condition}. Conversely, if \eqref{eq:inviscid condition}
holds, by integrating \eqref{eq: inviscid1} over $\mathbb{R}^{3}$
then $\frac{\partial}{\partial t}u(x,t)=0,$which yields that $u(x,t)=1$. 
\end{proof}

\section{Weakly homogeneous and weakly isotropic flows}

It has been pointed out in Secs. 2 and 3, the coefficients appearing
in the PDF PDE are determined by the conditional mean function $\mu(x,y,w,t)$
defined in (\ref{eq:mu-average}). The significance of this statistical
quantity of a turbulent flow has been demonstrated in the derivation
of the PDF PDE. By definition the conditional mean describes the deviation
from the local isotropicity of the turbulent flow, and therefore this
statistical characteristic can be used in the classification of turbulent
flows. In this section, we define weak homogeneous and weak isotropic
turbulent flows, then study the PDF equation for these turbulent flows.

The homogeneity and the isotropicity can be defined in general for
random fields indexed by a space variable $x\in\mathbb{R}^{d}$, which
has been introduced into the study of turbulence by G. I. Taylor \citep{Taylor1935}.
The local homogeneous and local isotropic flows were introduced by
Kolmogorov for formulating K41 theory (and its improved version K62
theory). According to Kolmogorov \citep{K41-b,K41-a}, a random field
$\{Z(x,t)\}$ is locally homogeneous if the conditional distribution
of $Z(y,t)-Z(x,s)$ given $Z(x,s)=z$ is independent of $(z,x,s)$
for any $t\geq s$, and further it is locally isotropic if the conditional
distribution of $Z(y,t)-Z(x,s)$ are invariant under reflections and
rotations. The usage of conditional mean function makes it possible
to generalize these concepts to their weak versions.

We may simplify the PDF PDE for turbulent flows when the vorticity
random field $W(x,t)$ is \emph{weakly homogeneous} in the sense that
for every pair $(y,x)$ and $t>0$, the conditional mean of $W^{i}(y,t)-W^{i}(x,t)$
given $W(x,t)=w$ depends only on $y-x$, $w$ and $t>0$ but independent
of $x$, so that we may write 
\[
\mu^{i}(x,y,w,t)=\beta^{i}(y-x,w,t),
\]
where $\beta^{i}(0,w,t)=0$. Then 
\[
a_{k}^{i}(x,w,t)=\frac{\partial\beta^{i}}{\partial x^{k}}(0,w,t),
\]
and 
\[
b^{i}(x,w,t)=\Delta\beta^{i}(0,w,t),
\]
which are still denoted by $a_{k}^{i}(w,t)$ and $b^{i}(w,t)$, though
they are independent of $x$. Similarly
\[
B^{i}(w,t)=\int_{\mathbb{R}^{3}}\frac{1}{4\pi|y|}\varepsilon^{ijk}\frac{\partial\beta^{k}}{\partial y^{j}}(y,w,t)\dd y,
\]
and 
\[
D^{i}(w,t)=w^{l}\int_{\mathbb{R}^{3}}\frac{y^{l}}{4\pi|y|^{3}}\varepsilon^{ijk}\frac{\partial\beta^{k}}{\partial y^{j}}(y,w,t)\dd y,
\]
which again are functions of $(w,t)$ only. Therefore the PDF PDE
can be simplified to be the following mixed type PDE

\[
\left(\frac{\partial}{\partial t}+B^{i}\frac{\partial}{\partial x^{i}}-\nu\Delta\right)f=\nu\frac{\partial}{\partial w^{i}}\left(a_{k}^{i}\frac{\partial f}{\partial x^{k}}\right)+\frac{\partial}{\partial w^{i}}\left(fA^{i}\right),
\]
where $A^{i}=D^{i}-2\nu b^{i}$ and $i=1,2,3$.

We say that the vorticity $W$ is\emph{ weakly isotropic}, if 
\[
\mathbb{E}\left[\nabla W(x,t)|W(x,t)=w\right]=0
\]
for all $x,w$ and $t$. By definition, $W$ is locally isotropic
in Kolmogorov's sense, then it is weakly isotropic.

If $W$ is weakly homogeneous and weakly isotropic, then $a=0$, $B$
and $A$ depend on $(w,t)$ only, so for this case, the PDF PDE becomes
a parabolic-transport equation 
\begin{equation}
\left(\frac{\partial}{\partial t}+B^{i}\frac{\partial}{\partial x^{i}}-\nu\Delta\right)f=\frac{\partial}{\partial w^{i}}\left(fA^{i}\right).\label{eq:par-trans1}
\end{equation}

After we have deduced the PDE \eqref{eq:par-trans1}, there is no
need to assume that $B$ and $A$ are independent of $x$. Therefore
we may consider the following PDE 
\begin{equation}
\left(\frac{\partial}{\partial t}+B^{i}\frac{\partial}{\partial x^{i}}+\frac{\partial B^{i}}{\partial x^{i}}-\nu\Delta\right)f=\frac{\partial}{\partial w^{i}}\left(fA^{i}\right).\label{eq:par-trans1-1}
\end{equation}
where both $B$ and $A$ are functions of three variables $x$, $w$
and $t$.

The following theorem provides the foundation for modeling weakly
isotropic turbulent flows based on the vorticity PDF. 
\begin{thm}
\label{thm:weak HI } Consider the parabolic-transport differential
equation \eqref{eq:par-trans1-1}, where $A$ and $B$ are considered
as given. Assume that $A(x,w,t)$ and $B(x,w,t)$ are Lipschitz continuous
in space variable $(x,w)$ uniformly in $t$. Suppose $f(w;x,t)$
is a solution to \eqref{eq:par-trans1-1} such that $fA$ is continuous
and decays to zero sufficiently fast as $|w|\to\infty$. 

1) If $f(w;x,0)\geq0$ for all $x$ and $w$, then $f(w;x,t)\geq0$
for all $t\geq0$, $x$ and $w$.

2) If $f(w;x,0)$ is a PDF for every $x,$ i.e. it is non-negative
and $\int_{\mathbb{R}^{3}}f(w;x,0)\dd w=1$ for every $x$, then so
are $f(w;x,t)$ for all $x$ and $t>0$ if and only if 
\begin{equation}
\frac{\partial}{\partial x^{i}}\int_{\mathbb{R}^{3}}B^{i}(x,w,t)f(w;x,t)\dd w=0\label{eq:iso-c2}
\end{equation}
for all $x$ and $t>0$.

3) Suppose in addition $B$ is independent of $w$ and $\frac{\partial B^{i}}{\partial x^{i}}=0$,
then if $f(w;x,0)$ is a PDF for every $x$, $f(w;x,t)$ is also a
PDF for every $x$ and for every $t>0$. 
\end{thm}

\begin{proof}
The proof is rather similar to that of Theorem \ref{thm:Invicid case},
so we only provide an outline. Let $(M_{t}^{i})_{t\geq0}$ ($i=1,2,3$)
be a three dimensional standard Brownian motion on a probability space
$(\varOmega,\mathcal{F},\mathbb{P})$. For every fixed $t>0$ we run
the following SDE system: 
\[
\dd X^{i}(t,s)=\sqrt{2\nu}\dd M_{s}^{i}-B^{i}(X(t,s),Y(t,s),t-s)\dd s,\quad X(t,0)=x,
\]
\[
\dd Y^{i}(t,s)=A^{i}(X(t,s),Y(t,s),t-s)\dd s,\quad Y(t,0)=w,
\]
which has a unique solution running up to $t$. Apply It\^o's formula
to 
\[
H_{s}=f(Y(t,s);X(t,s),t-s)e^{P(s)},
\]
where 
\[
P(s)=\int_{0}^{s}\frac{\partial A^{i}}{\partial w^{i}}(X(t,r),Y(t,r),t-r)\dd r.
\]
Then 
\begin{align*}
\dd H_{s} & =e^{P}\left(\frac{\partial f}{\partial w^{i}}\dd Y_{s}^{i}+\frac{\partial f}{\partial x^{i}}\dd X_{s}^{i}+f\frac{\partial A^{i}}{\partial w^{i}}-\frac{\partial f}{\partial s}\dd s+\nu\Delta f\dd s\right)\\
 & =\sqrt{2\nu}e^{P}\frac{\partial f}{\partial x^{i}}\dd M_{s}^{i}
\end{align*}
so that 
\[
H_{0}=H_{t}-\sqrt{2\nu}e^{P}\frac{\partial f}{\partial x^{i}}\dd M_{s}^{i}.
\]
By taking expectation, one may deduce that the solution can be expressed
as 
\[
f(w;x,t)=\mathbb{E}\left[f(Y(t,t);X(t,t),0)\exp\left(\int_{0}^{t}\frac{\partial A^{i}}{\partial w^{i}}(X(t,s),Y(t,s),t-s)\dd s\right)\right].
\]
Statements 1) - 3) can be easily shown using this expression and the
argument in the proof of Theorem \ref{thm:Invicid case}. If $B$
is independent of $w$, then, by integrating \eqref{eq:par-trans1},
\[
\left(\frac{\partial}{\partial t}+B^{i}\frac{\partial}{\partial x^{i}}-\nu\Delta\right)u(x,t)=0,
\]
where $u(x,t)=\int_{\mathbb{R}^{3}}f(w;x,t)\dd w$. Thus 3) follows
from the uniqueness of the solution of the previous parabolic equation. 
\end{proof}

\section{Modeling PDF of weakly isotropic flows}

In this and the next sections, several simple models based on our
PDF equation for modeling vorticity distributions in turbulence are
discussed. Clearly, the most straightforward way for modeling the
distribution of vorticity is to assign the conditional mean function
$\mu(x,y,w,t)$. The other parameters in the PDF PDE may be determined
accordingly.

For practical reasons which will be clarified in our computations
below, the simplest yet not trivial model for $\mu(x,y,w,t)$ should
be a function of $|y-x|^{2}$ only and $\mu$ decays sufficiently
fast at infinity. For such a model, the parameters $B^{i}$ and $D^{i}$
appearing in the PDF PDE \eqref{eq:vot-pdf1} vanish identically.
The other parameters $a_{k}^{i}$ and $b_{jk}^{i}$ are determined
by the asymptotic condition as $|y-x|\rightarrow0$. For this reason,
we assign
\begin{equation}
\mu^{i}(x,y,w,t)=\begin{cases}
-\frac{C}{6\nu}|y-x|^{2}, & \textrm{ if }|x-y|<\delta,\\
h(|y-x|^{2}), & \textrm{ if }|x-y|\geq\delta,
\end{cases}\label{eq:model-s5}
\end{equation}
for $i=1,2,3$, where $\delta>0$ and $h$ is a $\mathcal{C}^{1}$-function
such that $h$ and its derivative decay sufficiently fast at infinity.
$C\geq0$ is a model parameter which should be reflecting certain
physical properties of the underlying flows. The other parameters
can be read out from this model easily: $a_{k}^{i}=0$, $b_{jk}^{i}=-\delta_{jk}\frac{C}{6\nu}$
and therefore $b^{i}=-\frac{C}{2\nu}$. With this model of $\mu$,
the PDF PDE is simplified to be the following parabolic-transport
equation: 
\begin{equation}
\left(\frac{\partial}{\partial t}-\nu\Delta\right)f=C\sum_{i}\frac{\partial f}{\partial w^{i}},\label{eq:pdfpde-s5}
\end{equation}

The solution to \eqref{eq:pdfpde-s5} has a nice probabilistic representation
which may be read out from the proof of Theorem \ref{thm:weak HI }.
In fact 
\begin{equation}
f(w;x,t)=\mathbb{E}\left[f(Y(t,t);X(t,t),0)\right]\label{eq:s5-s1}
\end{equation}
where $f(w;x,0)$ is the PDF of the initial vorticity, and $X$ and
$Y$ are solutions to the SDE: 
\[
\begin{cases}
\dd X^{i}(t,s)=\sqrt{2\nu}\dd M_{s}^{i}, & X(t,0)=x,\\
\dd Y^{i}(t,s)=C\dd s, & Y(t,0)=w,
\end{cases}
\]
where $M$ is a standard 3D Brownian motion. The solutions when $s=t$
are given by 
\[
Y^{i}(t,t)=w^{i}+Ct
\]
and 
\[
X(t,t)=x+\sqrt{2\nu}M_{t}.
\]
Substituting these into \eqref{eq:s5-s1} we obtain 
\begin{align}
f(w;x,t) & =\mathbb{E}\left[f(w+Ct;x+\sqrt{2\nu}M_{t},0)\right]\nonumber \\
 & =\mathbb{E}\left[f(w+Ct;x+\sqrt{2\nu t}\xi,0)\right],\label{eq:s5-s2}
\end{align}
where $\xi$ is a random vector with the standard 3D normal distribution
$N(0,I_{3})$. Clearly the representation \eqref{eq:s5-s2} may also
be written in terms of Gaussian density 
\begin{align}
f(w;x,t) & =\frac{1}{(4\pi\nu t)^{3/2}}\int_{\mathbb{R}^{3}}f\left(w+Ct,z,0\right)e^{-\frac{|z-x|^{2}}{4\nu t}}\dd z\nonumber \\
 & =\frac{1}{(2\pi)^{3/2}}\int_{\mathbb{R}^{3}}f\left(w+Ct,x+\sqrt{2\nu t}z,0\right)e^{-\frac{|z|^{2}}{2}}\dd z.\label{eq:s5-s3}
\end{align}
Both representations \eqref{eq:s5-s2} and \eqref{eq:s5-s3} may be
used to evaluate $f(w;x,t)$ by using for example Monte-Carlo scheme
by sampling Gaussian random variables.

The initial PDF $f(w;x,0)$ is determined by the initial vorticity
distribution. For a turbulent flow, the initial vorticity may be written
as a sum: 
\[
W(x,0)=\omega_{0}(x)+\varepsilon(x),
\]
where $\omega_{0}(x)$ is the initial mean vorticity at the location
$x$, and $\varepsilon(x)$ represents a small random perturbation.
It is reasonable, therefore, to assume that $\varepsilon(x)$ has
a normal distribution $N(0,\sigma(x)^{2}I_{3})$, where the variance
$\sigma(x)$ may or may not depend on the location $x$. However,
if $\varepsilon(x)$ is allowed to depend on the location $x$, then
the noise $\varepsilon(x)$ has to satisfy the divergence-free condition
as well, a technical issue we will not address here in detail. It
is reasonable to assume that the initial vorticity mean $\omega_{0}$
is distributed in a small region of the space. In particular $\omega_{0}$
decays to zero sufficiently fast near infinity. Under this assumption,
\begin{equation}
f(w;x,0)=\frac{1}{\left(2\pi\sigma(x)^{2}\right)^{3/2}}\exp\left[-\frac{|w-\omega_{0}(x)|^{2}}{2\sigma(x)^{2}}\right].\label{eq:int-w}
\end{equation}
By substituting this into the representation \eqref{eq:s5-s2}, we
obtain 
\begin{equation}
f(w;x,t)=\mathbb{E}\left[\frac{1}{\left(2\pi\sigma\left(x+\sqrt{2\nu t}\xi\right)^{2}\right)^{3/2}}e^{-\frac{\left|w+Ct-\omega_{0}(x+\sqrt{2\nu t}\xi)\right|^{2}}{2\sigma\left(x+\sqrt{2\nu t}\xi\right)^{2}}}\right].\label{eq:mod-01}
\end{equation}
It is easy to see that $f(w;x,t)$ is no longer Gaussian except for
special cases which we would like to discuss below.

If $\omega_{0}$ vanishes and if $\sigma$ is a positive small constant,
the initial distribution is homogeneous and equation \eqref{eq:mod-01}
leads to a simple expression 
\[
f(w;x,t)=\frac{1}{\left(2\pi\sigma^{2}\right)^{3/2}}\exp\left(-\frac{\left|w+Ct\right|^{2}}{2\sigma^{2}}\right),
\]
which is independent of the location $x$ and remains a Gaussian density.
The interesting feature about this model is that the variance stays
as the constant $\sigma^{2}$ but new mean vorticity $-Ct$ is created
evenly after duration $t$. This is the case of a turbulence with
a small constant random perturbation. For turbulent flows observed
in nature, the initial vorticity mean $\omega_{0}$ does exist and
does not vanish, and the random noise $\varepsilon(x)$, for simplicity,
may be modeled by a Gaussian random variable independent of $x$.
Therefore when $\sigma>0$ is a small constant, we have 
\begin{equation}
f(w;x,t)=\frac{1}{\left(2\pi\sigma^{2}\right)^{3/2}}\mathbb{E}\left[e^{-\frac{|w+Ct-\omega_{0}\left(x+\sqrt{2\nu t}\xi\right)|^{2}}{2\sigma^{2}}}\right],\label{eq:mod-01-1}
\end{equation}
where $\xi\sim N(0,I_{3})$. 

For this case the mean vorticity $\omega(x,t)$ at $(x,t)$ can be
evaluated. Indeed 
\begin{align*}
\omega(x,t) & =\int_{\mathbb{R}^{3}}\frac{w}{\left(2\pi\sigma^{2}\right)^{3/2}}\mathbb{E}\left[e^{-\frac{|w+Ct-\omega_{0}\left(x+\sqrt{2\nu t}\xi\right)|^{2}}{2\sigma^{2}}}\right]\dd w\\
 & =\mathbb{E}\left[\int_{\mathbb{R}^{3}}\frac{w}{\left(2\pi\sigma^{2}\right)^{3/2}}e^{-\frac{|w+Ct-\omega_{0}\left(x+\sqrt{2\nu t}\xi\right)|^{2}}{2\sigma^{2}}}\dd w\right]\\
 & =\mathbb{E}\left[\omega_{0}\left(x+\sqrt{2\nu t}\xi\right)\right]-Ct,
\end{align*}
and therefore 
\begin{equation}
\omega^{i}(x,t)=\int_{\mathbb{R}^{3}}\frac{\omega_{0}^{i}(y)}{\left(4\pi\nu t\right)^{3/2}}e^{-\frac{\left|y-x\right|^{2}}{4\nu t}}dy-Ct.\label{eq:mean-wt}
\end{equation}
This equality shows that the vorticity mean $\omega(x,t)$ under this
simple model is independent of the noise parameter $\sigma^{2}$,
and $\omega(x,t)$ evolves according to the heat type equation 
\[
\frac{\partial}{\partial t}\omega(x,t)-C=\nu\Delta\omega(x,t),\textrm{ }\omega(x,0)=\omega_{0}(x),
\]
which is a rather crude approximation to the mean vorticity equation.

\section{Heat flow method}

In this section we propose another model for the PDF of the vorticity,
based on the heat flow method, in which the conditional mean function
$\mu$ is generated by a random field $R(x,\tau)$. The random field
$R(x,\tau)$ evolves according to the heat flow 
\begin{equation}
\left(\frac{\partial}{\partial\tau}-\nu\Delta\right)R^{i}(x,\tau)=0,\label{eq:flow-01}
\end{equation}
where $R(x,\tau)=(R^{1}(x,\tau),R^{2}(x,\tau),R^{3}(x,\tau))$. The
initial value $R(x,0)=\xi(x)=(\xi^{i}(x))$ is a centered Gaussian
noise white in space in the sense that the co-variance of $\xi^{i}$
at two locations $u$ and $v$ in $\mathbb{R}^{3}$ is given by 
\begin{equation}
\mathbb{E}\left[\xi^{i}(u)\xi^{i}(v)\right]=\delta(u-v).\label{eq:s-white}
\end{equation}
Also, we assume that $\xi^{i}$'s are i.i.d. random variables. For
simplicity, we assume that $R^{i}$ is centered here, but our argument
can definitely be generalized to the case when it is non-centered.

The solution to equation (\ref{eq:flow-01}) is given by 
\begin{equation}
R^{i}(x,\tau)=\int_{\mathbb{R}^{3}}\frac{1}{(4\pi\nu\tau)^{\frac{3}{2}}}e^{-\frac{|x-y|^{2}}{4\nu\tau}}\xi^{i}(y)\mathrm{d}y,\label{eq:flow-2}
\end{equation}
so that $R(x,\tau)$ is a centered Gaussian random field with its
co-variance 
\[
\begin{aligned}\sigma_{\tau}(x,y)= & \int_{\mathbb{R}^{3}}\int_{\mathbb{R}^{3}}\frac{1}{(4\pi\nu\tau)^{3}}e^{-\frac{|x-u|^{2}+|y-v|^{2}}{4\nu\tau}}\mathbb{E}\left[\xi^{i}(u)\xi^{i}(v)\right]\mathrm{d}u\mathrm{d}v\\
= & \int_{\mathbb{R}^{3}}\int_{\mathbb{R}^{3}}\frac{1}{(4\pi\nu\tau)^{3}}e^{-\frac{|x-u|^{2}+|y-v|^{2}}{4\nu\tau}}\delta(u-v)\mathrm{d}u\mathrm{d}v,\\
= & e^{-\frac{|x-y|^{2}}{8\nu\tau}}\int_{\mathbb{R}^{3}}\frac{1}{(4\pi\nu\tau)^{3}}e^{-\frac{|u-\frac{(x+y)}{2}|^{2}}{2\nu\tau}}\mathrm{d}u\\
= & \frac{1}{8(2\pi\nu\tau)^{\frac{3}{2}}}e^{-\frac{|x-y|^{2}}{8\nu\tau}}.
\end{aligned}
\]
It follows that the conditional distribution of $R(y,\tau)$ given
$R(x,\tau)=w$ has a normal distribution with mean $e^{-\frac{|x-y|^{2}}{8\nu\tau}}w$
and co-variance matrix 
\[
\frac{1}{8(2\pi\nu\tau)^{\frac{3}{2}}}\left(1-e^{-\frac{|x-y|^{2}}{4\nu\tau}}\right)I_{3}.
\]
Therefore the conditional mean of $R(y,\tau)-R(x,\tau)$ given $R(x,\tau)=w$
can be easily found to be 
\[
-\left(1-e^{-\frac{|x-y|^{2}}{8\nu\tau}}\right)w,
\]
which will be our $\mu(x,y,w,t)$ with $\tau=\varphi(t)$ a reparametrization
as part of the model. For simplicity let us consider the power law
model, that is, 
\begin{equation}
\tau=\lambda_{1}(t+\lambda_{2})^{\alpha},\label{eq:p-la1}
\end{equation}
where $\lambda_{1}$ and $\lambda_{2}$ are two positive numbers which
then become our model parameters. Since 
\[
\mu(x,y,w,t)=-\frac{|x-y|^{2}}{8\nu\tau}w+o(|x-y|^{2})w,
\]
therefore $a_{k}^{i}=0$ and $b_{jk}^{i}=-\delta_{jk}\frac{1}{8\nu\tau}w^{i}$.
In particular, 
\begin{equation}
b^{i}(w,t)=-\frac{3w^{i}}{8\nu\lambda_{1}(t+\lambda_{2})^{\alpha}}\label{eq:la2}
\end{equation}
for $i=1,2,3$. Since the conditional mean function $\mu$ depends
only on $|x-y|^{2}$ and decays exponentially fast at infinity, $B^{i}=0$
and $D^{i}=0$, which implies that 
\begin{equation}
A^{i}(w,t)=\frac{3w^{i}}{4\lambda_{1}(t+\lambda_{2})^{\alpha}}.\label{eq:la3}
\end{equation}
Thus the PDF PDE with these parameters is reduced to the simple parabolic-transport
equation 
\[
\left(\frac{\partial}{\partial t}-\nu\Delta_{x}\right)f=\frac{\partial}{\partial w^{i}}\left(A^{i}f\right).
\]
The divergence of $A$ is given by 
\[
\frac{\partial A^{i}}{\partial w^{i}}(x,w,t)=\frac{9}{4\lambda_{1}(t+\lambda_{2})^{\alpha}}.
\]

If $\alpha\neq1$ but $\alpha>0$, by using the Feynman-Kac formula
we have 
\begin{equation}
f(w;x,t)=\theta(\alpha,t)^{3}\mathbb{E}\left[f(Y(t,t),X(t,t),0)\right],\label{eq:m2-s0}
\end{equation}
where we have introduced 
\begin{equation}
\theta(\alpha,t)=\exp\left[\frac{3\left((t+\lambda_{2})^{1-\alpha}-\lambda_{2}^{1-\alpha}\right)}{4\lambda_{1}(1-\alpha)}\right]\label{eq:theta-d}
\end{equation}
for simplicity, and $X$ and $Y$ are solutions to SDE 
\[
\begin{cases}
\dd X^{i}(t,s)=\sqrt{2\nu}\dd M_{s}^{i}, & X(t,0)=x,\\
\dd Y^{i}(t,s)=\frac{3Y^{i}(t,s)}{4\lambda_{1}(t-s+\lambda_{2})^{\alpha}}\dd s, & Y(t,0)=w.
\end{cases}
\]

If $\alpha=1$ then 
\begin{equation}
f(w;x,t)=\lambda(t)^{3}\mathbb{E}\left[f(Y(t,t),X(t,t),0)\right],\label{eq:m2-s1-1}
\end{equation}
where 
\begin{equation}
\lambda(t)=\left(\frac{t+\lambda_{2}}{\lambda_{2}}\right)^{\frac{3}{4\lambda_{1}}}.\label{eq:ren-1}
\end{equation}

The previous SDEs have explicit solutions: 
\[
X(t,t)=x+\sqrt{2\nu}M_{t}
\]
and 
\[
Y(t,t)=w\theta(\alpha,t)
\]
if $\alpha\neq1$, and

\[
Y(t,t)=\lambda(t)w
\]
if $\alpha=1$.

Therefore, if $\alpha>0$ and $\alpha\neq1$, by plugging these explicit
solutions $X$ and $Y$ into \eqref{eq:m2-s0} we have 
\begin{align}
f(w;x,t) & =\theta(\alpha,t)^{3}\mathbb{E}\left[f\left(w\theta(\alpha,t),x+\sqrt{2\nu}M_{t},0\right)\right]\nonumber \\
 & =\theta(\alpha,t)^{3}\mathbb{E}\left[f\left(w\theta(\alpha,t),x+\sqrt{2\nu t}\xi,0\right)\right]\label{eq:m2-s1}
\end{align}
where $\xi\sim N(0,I_{3})$.

As in the previous section, the initial vorticity is assumed to be
of the form 
\[
W(x,0)=\omega(x)+\varepsilon(x),
\]
where $\varepsilon(x)$ has a normal distribution $N(0,\sigma^{2})$
and $\omega(x)$ is the mean vorticity at $x$, so that 
\[
f(w;x,0)=\frac{1}{\left(2\pi\sigma^{2}\right)^{3/2}}\exp\left[-\frac{\left|w-\omega(x)\right|^{2}}{2\sigma^{2}}\right]
\]
and therefore, according to \eqref{eq:m2-s1}, 
\[
f(w;x,t)=\theta(\alpha,t)^{3}\mathbb{E}\left[\frac{1}{\left(2\pi\sigma(x+\sqrt{2\nu t}\xi)^{2}\right)^{3/2}}e^{-\frac{\left|w\theta(\alpha,t)-\omega(x+\sqrt{2\nu t}\xi)\right|^{2}}{2\sigma(x+\sqrt{2\nu t}\xi)^{2}}}\right],
\]
where $\xi\sim N(0,I_{3})$. Hence $f(w;x,t)$ is no longer Gaussian
density unless $\sigma(x)$ is independent of $x$ and $\omega=0$
identically.

Let us discuss a special case where $\omega=0$ identically and $\sigma>0$
is constant. Then 
\[
f(w;x,t)=\frac{\theta(\alpha,t)^{3}}{\left(2\pi\sigma^{2}\right)^{3/2}}\exp\left[-\frac{\left|w\theta(\alpha,t)\right|^{2}}{2\sigma^{2}}\right]
\]
is a Gaussian density with mean $0$ and variance $\rho^{2}I_{3}$,
where 
\begin{equation}
\rho^{2}(x,t)=\frac{1}{\sqrt{\theta(\alpha,t)}}\sigma^{2}.\label{eq:var-001}
\end{equation}

If $\sigma>0$ is a constant and $\omega$ does not vanish identically,
then $f(w;x,t)$ is not Gaussian, but has a nice representation: 
\begin{equation}
f(w;x,t)=\frac{\theta(\alpha,t)^{3}}{\left(2\pi\sigma^{2}\right)^{3/2}}\mathbb{E}\left[e^{-\frac{\left|w\theta(\alpha,t)-\omega(x+\sqrt{2\nu t}\xi)\right|^{2}}{2\sigma^{2}}}\right],\label{eq:rep-w01}
\end{equation}
where $\xi\sim N(0,I_{3})$, and $\theta(\alpha,t)$ is defined by
\eqref{eq:theta-d}.

A similar discussion applies to the model where $\alpha=1$, which
is certainly interesting too. For this model the PDF $f(w;x,t)$ is
given by 
\begin{equation}
f(w;x,t)=\lambda(t)^{3}\mathbb{E}\left[f\left(w\lambda(t),x+\sqrt{2\nu t}\xi,0\right)\right]\label{eq:a0re-01}
\end{equation}
where $\xi\sim N(0,I_{3})$ and $\lambda(t)$ is defined by \eqref{eq:ren-1}.

Suppose again the initial vorticity 
\[
W(x,0)=\omega(x)+\varepsilon(x),
\]
where $\varepsilon(x)$ has a normal distribution $N(0,\sigma^{2}(x))$
and $\omega(x)$ is the mean. Then 
\begin{equation}
f(w;x,t)=\mathbb{E}\left[\frac{\lambda(t)^{3}}{\left(2\pi\sigma(x+\sqrt{2\nu t}\xi)^{2}\right)^{3/2}}e^{-\frac{\left|w\lambda(t)-\omega\left(x+\sqrt{2\nu t}\xi\right)\right|^{2}}{2\sigma(x+\sqrt{2\nu t}\xi)^{2}}}\right],\label{eq:g-int-01}
\end{equation}
which is not Gaussian except for the following special case.

If $\sigma(x)=\sigma>0$ is a constant and $\omega(x)=0$, then 
\[
f(w;x,t)=\frac{\lambda(t)^{3}}{\left(2\pi\sigma^{2}\right)^{3/2}}e^{-\frac{\left|w\lambda(t)\right|^{2}}{2\sigma^{2}}}
\]
is Gaussian with mean zero and variance $\rho^{2}I_{3}$, where 
\[
\rho^{2}(x,t)=\sigma^{2}\left(\frac{\lambda_{2}}{t+\lambda_{2}}\right)^{\frac{3}{2\lambda_{1}}}.
\]

The most interesting case for the purpose of modeling turbulent flows
is the model where $\sigma>0$ is a small parameter and $\omega$
is not zero, so that 
\[
f(w;x,t)=\frac{\lambda(t)^{3}}{(2\pi\sigma^{2})^{3/2}}\mathbb{E}\left[e^{-\frac{\left|w\lambda(t)-\omega\left(x+\sqrt{2\nu t}\xi\right)\right|^{2}}{2\sigma^{2}}}\right]
\]
where $\xi\sim N(0,I_{3})$. Although it is no longer Gaussian, some
of its features can be extracted by doing Monte-Carlo simulations.

We may conclude that this model provides some nice features of the
propagation of the vorticity which may be helpful for the understanding
of the energy dissipation in turbulence, and we will explore this
in a separate work.


\begin{thebibliography}{99}
\bibitem{Batchelor1953} G.~K. Batchelor. \newblock {\em {The
Theory of Homogeneous Turbulence}}. \newblock Cambridge University
Press, 1953.

\bibitem{Chorin1994} A.~J. Chorin. \newblock {\em {Vorticity
and Turbulence}}. \newblock Springer Science \& Business Media,
2013.

\bibitem{CottetandKoumoutsakos2000} G.-H. Cottet and P.~D. Koumoutsakos.
\newblock {\em {Vortex Methods: Theory and Practice}}. \newblock
Cambridge University Press, 2000.

\bibitem{Davidson2004} P. A. Davidson. \newblock {\em {Turbulence:
An Introduction for Scientists and Engineers}}. \newblock Oxford
University Press, 1995.

\bibitem{Frisch1995} U.~Frisch. \newblock {\em {Turbulence:
The Legacy of A. N. Kolmogorov}}. \newblock Cambridge University
Press, 1995.

\bibitem{Hopf1952} E.~Hopf. \newblock {Statistical hydromechanics
and functional calculus}. \newblock {\em Journal of Rational Mechanics
and Analysis}, $\mathbf{1}$:87--123, 1952.

\bibitem{K61} A.~N. Kolmogorov. \newblock A refinement of previous
hypotheses concerning the local structure of turbulence in a viscous
incompressible fluid at high Reynolds number. \newblock {\em Journal
of Fluid Mechanics}, $\mathbf{13}$(1):82--85, 1962.

\bibitem{K41-a} A.~N. Kolmogorov. \newblock The local structure
of turbulence in incompressible viscous fluid for very large {R}eynolds
numbers. \newblock {\em Proceedings of the Royal Society of London.
Series A: Mathematical and Physical Sciences}, $\mathbf{434}$(1890):9--13,
1991.

\bibitem{K41-b} A.~N. Kolmogorov. \newblock {Dissipation of energy
in the locally isotropic turbulence}. \newblock {\em Proceedings
of the Royal Society of London. Series A: Mathematical and Physical
Sciences}, $\mathbf{434}$(1890):15--17, 1991.

\bibitem{Lesieur1997} M.~Lesieur. \newblock {\em {Turbulence
in Fluids}}. \newblock Springer Science \& Business Media, 1997.

\bibitem{MajadaBertozzi2002} A.~J. Majda and A.~L. Bertozzi. \newblock
{\em {Vorticity and Incompressible Flow}}. \newblock Cambridge
University Press, 2002.

\bibitem{Moffatt2011} H.~K. Moffatt. \newblock {A brief introduction
to vortex dynamics and turbulence}. \newblock In {\em {Environmental
Hazards: The Fluid Dynamics and Geophysics of Extreme Events}},
pages 1--27. World Scientific, 2011.

\bibitem{MoninandYaglom1965} A.~Monin and A.~Yaglom. \newblock
{\em {Statistical Fluid Mechanics: Mechanics of Turbulence}},
volume 1 and 2. \newblock MIT Press, 1971 and 1975.

\bibitem{Pope1985} S.~B. Pope. \newblock PDF methods for turbulent
reactive flows. \newblock {\em Progress in Energy and Combustion
Science}, $\mathbf{11}$(2):119--192, 1985.

\bibitem{Pope2000} S.~B. Pope. \newblock {\em {Turbulent Flows}}.
\newblock Cambridge University Press, 2000.

\bibitem{rozanov1982} Y.~A. Rozanov. \newblock {\em Markov Random
Fields}. Springer, 1982.

\bibitem{Saffman1992} P.~G. Saffman. \newblock {\em {Vortex
Dynamics}}. \newblock Cambridge University Press, 1992.

\bibitem{Taylor1921} G.~I. Taylor. \newblock {Diffusion by continuous
movements}. \newblock {\em Proceedings of the London Mathematical
Society}, $\mathbf{s2}$-$\mathbf{20}$(1):196--212, 1921.

\bibitem{Taylor1935} G.~I. Taylor. \newblock {Statistical theory
of turbulence, Parts 1-5}. \newblock {\em Proceedings of the Royal
Society A: Mathematical, Physical and Engineering Sciences}, $\mathbf{151}$(873):421--478,
1935. 
\end{thebibliography}
\end{document}